\documentclass[12pt,onecolumn]{IEEEtran}
\usepackage{amsfonts}
\usepackage{hyperref}
\usepackage{fancyhdr}
\usepackage[all]{xy}
\usepackage{mathrsfs}
\usepackage[latin1]{inputenc}
\usepackage[usenames,dvipsnames]{color}
\usepackage{makeidx}
\usepackage{graphicx}
\usepackage{amsthm}
\usepackage{amssymb}
\usepackage{amsmath}

\makeindex

\newcommand{\notransversal}{\mathbin{\hbox{$\cap$ \kern-.45cm
\raise.3ex\hbox{$\top$}\kern-.26cm \raise.3ex\hbox{$/$}}}}

\newtheorem{theorem}{Theorem}
\newtheorem{corollary}[theorem]{Corollary}
\newtheorem{lemma}[theorem]{Lemma}
\newtheorem{proposition}[theorem]{Proposition}
\newtheorem{example}[theorem]{Example}
\newtheorem{definition}[theorem]{Definition}
\newtheorem{remark}[theorem]{Remark}

\vspace{25cm}

\begin{document}
\thispagestyle{empty}

\pagenumbering{roman} 

\pagestyle{headings}

\title{Abelian Noncyclic Orbit Codes and Multishot Subspace Codes}
\author{Gustavo Terra Bastos, Reginaldo Palazzo J\'unior, and Marin\^es Guerreiro \thanks{G.T. Bastos
is with the Department of Mathematics and Statistics, UFSJ, Brazil. R.
Palazzo Jr. is with the Department of Communications, FEEC, Universidade
Estadual de Campinas, UNICAMP, Brazil. M. Guerreiro
is with the Department of Mathematics, UFV, Brazil. This work has been
supported by FAPESP under grant 2013/25977-7, CNPq under grant
305656/2015-5. E-mails: gtbastos@ufsj.edu.br; palazzo@dt.fee.unicamp.br; marines@ufv.br}}

\maketitle

\begin{abstract}
In this paper we characterize the orbit codes as geometrically uniform codes. This characterization is based on the description of all isometries over a projective geometry. In addition, the Abelian orbit codes are defined and a new construction of Abelian non-cyclic orbit codes is presented. In order to analyze their structures, the concept of geometrically uniform partitions have to be reinterpreted. As a consequence, a substantial reduction in the number of computations needed to obtain the minimum subspace distance of these codes is achieved and established. 

An application of orbit codes to multishot subspace codes obtained according to a multi-level construction is provided.
\end{abstract}

Geometrically uniform codes, \and Abelian orbit codes, \and Multishot subspace codes, \and Geometrically uniform partitions.

\section{Introduction}

In a communication network system, the overall throughput of the network can be increased with the introduction of the concept of Network Coding~\cite{ahl}. In order to correct possible errors or erasures during a transmission, one of the proposed coding techniques to be employed is the class of \textit{Subspace Codes}~\cite{koetterk}. The strategy behind these codes may be described as follows: an information packet sent by the source, processed in the intermediate nodes, and received at the destinations, can be seen as a ``codeword" or as a ``point", or equivalently, as a vector subspace in a projective space.

Let $\mathcal{P}_q (n)$ be a projective space. Thus, $\mathcal{P}_q (n)$ denotes the set of all vector subspaces of a vector space $\mathbb{F}_{q}^n$ or $\mathbb{F}_{q^n}$, with $\mathbb{F}_q$ a finite field with $q$ elements, for $q$ a prime number or a power of a prime. This set is partitioned into subsets, each one called a \textit{Grassmannian} and denoted by $\mathcal{G}_q (n,k)$. They are defined as the collection of all $k$-dimensional subspaces of  $\mathbb{F}_{q}^n$, with cardinality $\left|\mathcal{G}_q (n,k)\right|=\left[
\begin{array}{c}
 n \\
 k \\
 \end{array}
\right]_{q}$, the Gaussian binomial coefficient. An $(n,\textbf{M},d)$-\emph{subspace code} $C$ is a collection of $\textbf{M}$ vector subspaces from $\mathbb{F}_q ^n$ with minimum distance $d$. In particular, if every codeword in $C$ has the same dimension $k$ $(0\leq k\leq n)$, then $C$ is an $(n,\textbf{M},d,k)$-constant dimension subspace code, or simply a \emph{constant dimension code}. The minimum distance $d=d_S (C)$ is computed using a metric called \emph{subspace distance}. Given two distinct codewords $V,W \in C$, the distance between them is $d_S (V,W)=$ dim $V$ + dim$W$ - 2dim$(V \cap W )$. When $C$ is a constant dimension code, $d$ will be an even number. We refer the reader to~\cite{kha} for more detailed information regarding subspace codes employed in the Network Coding context.

Denote by $GL_n \left(\mathbb{F}_q \right)$ the general linear group of $n\times n$ invertible matrices over a
finite field  $\mathbb{F}_q$.
In~\cite{orbitcodesnew} it is proposed a new way to describe constant dimension codes. They are called \emph{orbit codes} when it is considered sn action of a subgroup $G$ of $GL_n \left(\mathbb{F}_q \right)$
 on a $k$-dimensional subspace $V$. If $G$ is a cyclic group, the subspace code is called a \emph{cyclic orbit code}.

A second approach describing cyclic orbit codes is presented in~\cite{cyclic}. The difference between these two approaches is that, in the first one, $\mathcal{P}_q (n)$ is considered as a collection of subspaces from $\mathbb{F}_q ^n$, whereas in the second $\mathcal{P}_q (n)$ is considered as a collection of subspaces from $\mathbb{F}_{q^n}$. Actually, these two approaches are the same due to the vector-space isomorphism between $\mathbb{F}_q ^n$ and $\mathbb{F}_{q^n}$. However, each approach has its advantages regarding the construction of cyclic orbit codes.

Recently, new constructions of orbit codes were proposed. For instance, in~\cite{castlemeeting} it is proposed a construction of orbit codes using as generating group a more general Abelian group, whose codes attain the maximum subspace distance. In~\cite{iran}, a well-structured construction of non-Abelian orbit codes is proposed making use of the semi-direct product of the cyclic group generated by a primitive element of $\mathbb{F}_{q^n}$ with the cyclic group generated by the Frobenius automorphism. Due to the transitive action of the corresponding algebraic structure, these orbit codes can be classified as geometrically uniform subspace codes. On the other hand, in~\cite{uniaoantigo} and more recently in~\cite{china}, constructions of constant dimension codes based on the union of cyclic orbit codes with a prescribed minimum distance are provided. In this case, it is not possible to state that these constant dimension subspace codes are orbit codes and, consequently, geometrically uniform subspace codes.

The previous constructions motivates us to consider the important aspect that the class of geometrically uniform subspace codes is the proper class for the identification of orbit codes. The classical geometrically uniform codes, as proposed by Forney~\cite{forney}, can be seen as classical orbit codes since there is a group action on a codeword (vector). By making use of this concept to the subspace code context together with the classification of isometries in projective geometry as shown in~\cite{isometriasana}, we characterize all geometrically uniform subspace codes as orbit codes. Since orbit codes may be seen as geometrically uniform subspace codes it follows that the classical definitions and results about geometrically uniform codes may be generalized to the case in consideration. In particular, using the geometrically uniform partitions, Theorem~\ref{theoremprincioal} is established showing that the number of computations necessary to obtain the minimum subspace distance of Abelian orbit codes may be reduced substantially. As a consequence, its corollary provides the exact number of computations necessary to obtain the minimum subspace distance of cyclic orbit codes. Finally, an application of the orbit codes to a construction of multishot subspace codes~\cite{nobrega} is provided, where it is possible to note some advantages in their implementation, which comes from the geometrically uniform properties.

It is worth mentioning that the identification of orbit codes as geometrically uniform subspace codes does not imply in a new subspace code construction, it just provides a new way of viewing this class of codes so that one may fully explore its inherent geometric and algebraic properties.

This paper is organized as follow: In Section \ref{section2}, the definitions of orbit codes and, in particular, cyclic orbit codes are presented. We also propose a new construction of Abelian non-cyclic orbit codes such that, for fixed $n$, $q$, $d$ and $q-1 \geq n$, it meets the best cardinality between the orbit codes known so far. In Section \ref{codsubgeomunif}, we review some of the basic concepts of geometrically uniform codes~\cite{forney}, where the corresponding results will be adapted to the context of subspace codes.
From the characterization of geometrically uniform subspace codes as orbit codes and from some definitions in~\cite{biglieri}, we provide a procedure to reduce the number of computations of the minimum subspace distance by focusing on the Abelian orbit codes.
 In Section~\ref{multishot}, given $S \subseteq \mathcal{G}_q (n,k)$ an alphabet for multishot codes construction, we propose to use group action given by $G< GL_n \left(\mathbb{F}_q\right)$ over $S$ in order to partition it as a collection of orbit codes. Applying the subgroups of $G$ over these orbit codes and their respective orbit subcodes, we provide a systematic way of partitioning $S$ and a considerable reduction of the number of computations needed to obtain the intrasubset subspace distance in each level of the partition. Finally, in Section \ref{section4} the conclusions are drawn.

\section{Orbit Codes in $\mathcal{G}_q (n,k)$}\label{section2}

Trautmann, Manganiello and Rosenthal~\cite{orbitcodesnew} have introduced the concept of orbit codes in the network coding context. This class of codes is generated by a subgroup $G$ of 
$GL_n \left(\mathbb{F}_q \right)$ acting on a $k$-dimensional vector subspace of the vector space $\mathbb{F}_q ^n$. If $G$ is an Abelian group, then the code is said to be an \emph{Abelian orbit code}. In particular, if we take a cyclic group of $GL_n \left(\mathbb{F}_q \right)$, then this code is said to be a \emph{cyclic orbit code}. Considering the latter case, an alternative definition is the one based on the vector subspaces of a finite field $\mathbb{F}_{q^n}$, with an specific subspace taken as the initial ``point" and the action of the cyclic group on such a ``point", which results in a cyclic orbit code.

\begin{definition}~\cite{coc}\label{defdecodcicana}
Let $G$ be a subgroup of $GL_n \left(\mathbb{F}_q \right)$. Then $C_G (V)=\left\{rs(\mathcal{VA}): \mathcal{A} \in G \right\}$ is called an \emph{orbit code}, with $V$ a $k$-dimensional vector space in $\mathcal{G}_q (n,k)$, and $rs(\mathcal{A})$ denotes the row space generated by the matrix $\mathcal{A}$. In particular, for $G=\langle \mathcal{A}\rangle$ a cyclic group, $C_{\langle \mathcal{A}\rangle} (V)$ is said to be a \emph{cyclic orbit code}.
\end{definition}

\begin{definition}\cite{coc}\label{definitiondeirredutivel}
A matrix $\mathcal{A}\in GL_n \left(\mathbb{F}_q \right)$ is irreducible if $\mathbb{F}_q ^n$ contains no nontrivial $\mathcal{A}$-invariant subspace, otherwise it is reducible. A subspace $V\subseteq \mathbb{F}_q ^n$ is $\mathcal{A}$-invariant if $rs(\mathcal{VA})=V$. A non-trivial subgroup $G\leq GL_n \left(\mathbb{F}_q \right)$ is irreducible if $\mathbb{F}_q ^n$ contains no non-trivial $G$-invariant subspace, otherwise it is reducible.
\end{definition}

Let $\alpha \in \mathbb{F}_{q^n}$ be a root of an irreducible polynomial $p(x)\in \mathbb{F}_q [x]$, with degree $n$ a positive integer. If $\mathbb{F}_{q^n}$ is seen as an $\mathbb{F}_q$-vector space, then  the following isomorphisms of vector spaces hold

\begin{equation}\label{representacoesdefqn}
\mathbb{F}_{q^n} \simeq \mathbb{F}_q [x]/\langle p(x)\rangle \simeq \mathbb{F}_q [\alpha]\simeq \mathbb{F}_q ^n .
\end{equation}
In other words, the vector space $\mathbb{F}_{q}^{n}$ may be realized at least in these three distinct ways.
In particular, if $p(x)$ is a primitive polynomial, then $\alpha$ is a primitive element of $\mathbb{F}_{q^n}$ and a $k$-dimensional vector subspace $V$ of $\mathbb{F}_{q^n}$, for $1\leq k \leq n$, is denoted as follows

\begin{equation}\label{descricaoespvet}
V=\left\{0,\alpha^{i_1} , \alpha^{i_2} ,..., \alpha^{i_{q^k -1}}\right\}.
\end{equation}

\begin{remark}
From now on, $\alpha$ will denote a primitive element (a root of $p(x)$) in \linebreak$\displaystyle{\mathbb{F}_{q^{n}}\simeq \mathbb{F}_q [x]/\langle p(x)\rangle}$ and the isomorphisms shown in \eqref{representacoesdefqn} will be used freely.
\end{remark}

Hence we may define cyclic orbit codes also as follows.

\begin{definition}\label{troya}\cite{cyclic}
Fix an element $\beta= \alpha^j$ of $\mathbb{F}_{q^n} ^{*} \setminus \{1\}$. Let $V$ be a subspace of the vector space $\mathbb{F}_{q^n}$. The $\beta$-\emph{cyclic orbit code} generated by $V$ is defined as the set
\begin{equation}
C_{\langle \beta\rangle} (V):=\left\{V \beta^i : i=0,1,..., \mbox{ord}(\beta)-1\right\}.
\end{equation}

If $\beta=\alpha$ or $\beta$ is equal to any other primitive element of $\mathbb{F}_{q^n}$, then the $\beta$-cyclic orbit code is denoted by $C_{\langle\alpha \rangle} (V)$ and it is called a cyclic orbit code.
\end{definition}

\begin{example}
Let $p(x)=x^6 + x+1$ be a primitive polynomial in $\mathbb{F}_2 [x]$ and $\alpha \in \mathbb{F}_{2^6}\simeq \mathbb{F}_2 [x] / \langle p(x)\rangle$ a root of $p(x)$. Given $V=\left\{0, \alpha, \alpha^8 , \alpha^{12}, \alpha^{26}, \alpha^{27}, \alpha^{32}, \alpha^{35}\right\}$ a $3$-dimensional vector subspace of $\mathbb{F}_{2^6}$, then the cyclic orbit code $C_{\langle \alpha \rangle} (V)$ is a $(6,63,4,3)$-constant dimension code.
\end{example}

\begin{definition}\label{estabilizador}
Let $C_G(V)$ be an orbit code. The \emph{stabilizer} of a $k$-dimensional vector subspace $V$ of $\mathbb{F}_q ^n$ is the subgroup $Stab_G (V):=\left\{\mathcal{A} \in G : rs(\mathcal{VA})=V \right\}$ of $G$. If $V$ is seen as a $k$-dimensional vector subspace of $\mathbb{F}_{q^n}$ and $\langle \beta \rangle = G$, for $\beta=\alpha^j $, then, by abuse of notation, $Stab_{G} (V):=\left\{\beta^i \in G : V \beta^i = V \right\}$.
\end{definition}

\begin{definition}\label{spread}
Given $r$ a positive integer such that $r|n$ and an $r$-dimensional vector subspace $V=\mathbb{F}_{q^r}$ of $\mathbb{F}_{q^n}$, 
the orbit code $C_{\langle \alpha \rangle} (V)$ is called an \emph{spread code}.
\end{definition}

Note that an spread code is a $\left(n, \frac{q^n -1} {q^r -1},2r,r\right)$-constant dimension code and it is an example of an optimal subspace code. For more information about spread codes, we refer the reader to~\cite{spreadcodes}.
%
%

In~\cite{cyclic} and \cite{coc}, it is observed that the minimum distance of an orbit code $C_{G} (V)$ can be obtained by listing all the subspace distances between a given codeword $V$ and the remaining ones, since
\begin{eqnarray}
d_S \left(rs(\mathcal{VA}), rs(\mathcal{VB}) \right)&=&d_S \left( V, rs\left(\mathcal{VBA}^{-1}\right) \right)\mbox{ or }\nonumber\\
d_S \left(V\alpha^i, V\alpha^j) \right)&=&d_S \left( V, V\alpha^{j-i}\right),
\end{eqnarray}
for any $\mathcal{A},\mathcal{B} \in G$. This is due to the fact that the elements of $GL_n \left(\mathbb{F}_q \right)$ (and powers of $\alpha$) act as isometries on $\mathcal{P}_q (n)$, see Section~\ref{codsubgeomunif}. Thus, the minimum (subspace) distance of this class of codes is given by
\begin{equation}
d=d_S \left(C_G (V)\right)=\min_{g \in G \setminus Stab_G (V)} \left\{d_S (V, Vg)\right\}.
\end{equation}

To the best of our knowledge, the first Abelian non-cyclic orbit codes with parameters $(n,q(q-1),2k,k)$,~\cite[Theorem 3]{castlemeeting}, were proposed in~\cite{castlemeeting}, satisfying the inequalities

\begin{equation}\label{restricaocoddoclement}
p^{r-1}\leq n-2k<k<n-k\leq p^r=q,
\end{equation}
where $r\geq 1$, $p$ prime, $n$ and $k$ positive integers.

\subsection{A New Construction of Abelian non-Cyclic Orbit Codes}

Given $A=\left[a_{ij} \right] \in GL_n (q)$, let $UT_n (q) < GL_n (q)$ be the non-Abelian group of the upper triangular matrices, that is,
\begin{equation}
UT_n (q):=\left\{ A \in GL_n (q) : a_{ij}=0 \mbox{ for } i<j \right\}.
\end{equation}

It is known that $\left|UT_n (q)\right|=q^M$, for $\displaystyle{M=\sum_{i=1} ^{n-1} i}$,
and
\begin{equation}
\left|GL_n (q)\right|=\left|UT_n (q) \right| \cdot m, \mbox{ with } m=\prod_{j=1} ^n q^j -1.
\end{equation}

Therefore, $UT_n (q) $ is a $p$-Sylow subgroup of $GL_n(q)$~\cite{isaacs}.

%

The largest Abelian subgroup of $UT_n  (q)$, for $q$ odd, is described by the following theorem.

\begin{theorem}\cite{abeliano}\label{maxorder}
Let $\mathbb{F}_q$ be a finite field of order $q=p^t$ ($p$ an odd prime). The maximal order of an Abelian $p$-subgroup of $GL_n (q)$ is $q^{\left\lfloor \frac{n^2}{4}\right\rfloor}$ and this maximum is attained.
\end{theorem}

In the proof of Theorem~\ref{maxorder}, the author states that the group consisting of the matrices of the form
\begin{equation}\label{grupoabelianomaximal}
\overline{G}:=\left\{\left[\begin{array}{cc}
    Id_{n-r} & H_r  \\
    0_{r} & Id_{n-r}
\end{array}\right]: H_r \in \mathbb{F}_q ^{r \times r}\right\}\mbox{, with }r=\left\lfloor \frac{n}{2}\right\rfloor,
\end{equation}
attains the maximal order.

Theorem~\ref{maxorder} establishes that the largest Abelian $p$-subgroup is obtained for $p$ an odd prime and $r=\left\lfloor \frac{n}{2}\right\rfloor$. In spite of the fact that this subgroup is not the largest one for $p=2$, we consider this possibility in the construction of orbit codes, since it is still possible to obtain orbit codes with large cardinality. From now on, we consider $n=2r$.
%

Let $V$ be a $k$-dimensional vector subspace of $\mathbb{F}_q ^n $ such that $V=rs(\mathcal{V})$, with 
\begin{equation}\label{representacaodeespacovetorial}
\mathcal{V}=
\left[\begin{array}{ll}
A_{l \times r}     & B_{l \times r}  \\
C_{(k-l) \times r}     & D_{(k-l) \times r}
\end{array}\right].
\end{equation}
If $G= \left[\begin{array}{cc}
    Id_{r} & H_r  \\
    0_r &  Id_r
\end{array}\right]\in \overline{G}$, then

\begin{eqnarray*}
rs(\mathcal{V})&=&rs(\mathcal{V}G) \Leftrightarrow \left[\begin{array}{ll}
A_{l \times r}     & B_{l \times r}  \\
C_{(k-l) \times r}     & D_{(k-l) \times r}
\end{array}\right] = \left[\begin{array}{ll}
A_{l \times r}     &A_{l \times r}H_r + B_{l \times r} \\
C_{(k-l) \times r}     &C_{(k-l) \times r}H_r +  D_{(k-l) \times r}
\end{array}\right]\nonumber \\
&&\nonumber \\
&\Leftrightarrow& A_{l \times r}H_r =0_{l \times r} \mbox{ and } C_{(k-l) \times r}H_r =0_{(k-l) \times r}\nonumber
\end{eqnarray*}

Denoting by $\widetilde{G}:=\left\{H_r \in \mathbb{F}_q ^{r \times r} : A_{l \times r}H_r =0_{l \times r} \mbox{ and } C_{(k-l) \times r}H_r =0_{(k-l) \times r}\right\} $, we get
\begin{equation}
\left|C_{\overline{G}} (V) \right|=\frac{|\overline{G}|}{\left|Stab_{\overline{G}} (V)\right|}=\frac{|\overline{G}|}{\left|\widetilde{G} \right|}.
\end{equation}

For the proposed construction, we may compute a bound for the minimum subspace distance of an Abelian non-cyclic orbit code $C_{\overline{G}} (V)$ according to the rank $(rk)$ of the submatrices which are part of the matrices of $G$ and the matrix whose row space is $V$.

\begin{theorem}\label{distanciaminimaviarank}
Let $G=\left[\begin{array}{cc}
    Id_r &  H_r \\
    0_r  &  Id_r
\end{array}\right] \in \overline{G}$ and $V$ be a $k$-dimensional vector subspace of $\mathbb{F}_q ^n$ as described in~\eqref{representacaodeespacovetorial}, then
\begin{equation}
d_S \left(rs(\mathcal{V}), rs(\mathcal{V}G) \right)\leq  2rk\left(\left[\begin{array}{c}
     A_{l \times r}H_{r}  \\
     C_{(l-r) \times r}H_{r},
\end{array}\right]\right).
\end{equation}

If $k=r=\frac{n}{2}$ and $\mathcal{V}=\left[\begin{array}{cc}
    Id_ r & A_r
\end{array}\right]$, then $d_S \left(rs(\mathcal{V}), rs(\mathcal{V}G) \right)=2 rk(\left[H_r \right])$.
\end{theorem}

\begin{proof}
\begin{eqnarray*}
d_S \left(rs(\mathcal{V}), rs(\mathcal{V}G) \right)&=& \mbox{dim} (rs(\mathcal{V}))+\mbox{dim} (rs(\mathcal{V}G)) -2\mbox{dim} (rs(\mathcal{V}) \cap rs(\mathcal{V}G) )\\
&=& 4k-2\mbox{dim} (rs(\mathcal{V}) \cap rs(\mathcal{V}G) ) -2k \\&=&2\mbox{dim} (rs(\mathcal{V})+ rs(\mathcal{V}G)) -2k\\
&=&2 rk\left(\left[\begin{array}{cc}
    A_{l \times r} & B_{l \times r} \\
    C_{(k-l) \times r} & D_{(k-l) \times r} \\
    A_{l \times r}  & A_{l\times r }H_{r} +B_{l \times r}\\
    C_{(k-l) \times r} & C_{(l-k)\times r }H_{r} +D_{(k-l) \times r}
\end{array}\right]\right) - 2k\\
&\leq & 2rk\left(\left[\begin{array}{c}
     A_{l \times r}H_{r}  \\
     C_{(l-r) \times r}H_{r},
\end{array}\right]\right),
\end{eqnarray*}
where the inequality comes from $rk\left(\left[\begin{array}{c}
     X  \\
     Y
\end{array}\right]\right) \leq rk\left(Y-X\right)+\min\{rk(X),rk(Y)\}.
$

In particular, if $k=r$ and $\mathcal{V}=\left[Id_r \,\, A_r \right]$, then $\mathcal{V}G = \left[Id_r \,\,\, (H + A)_r  \right]$ and
\begin{eqnarray*}
d_S \left(rs(\mathcal{V}), rs(\mathcal{V}G) \right)&=& 
2 rk\left(\left[\begin{array}{cc}
    Id_{r} & A_{ r} \\
    Id_{r} & H_r + A_{r}
\end{array}\right]\right) - 2k\\
&=& 2 rk \left(\left[\begin{array}{cc}
       Id_{r} & A_{ r} \\
   0_r & H_r
\end{array}\right]\right) - 2k\\
&= &2 rk \left(H_{r}\right).
\end{eqnarray*}
\end{proof}

\begin{remark}
If $k=r$, $\mathcal{V}=\left[Id_r \,\, A_r \right]$ and $\mathcal{V}G = \left[Id_r \,\,\, (H + A)_r  \right]$, then $V=VG$ if, and only if, $H_r = 0_r$, namely, $C_{\overline{G}} (V)$ has trivial stabilizer.
\end{remark}

Besides the results obtained so far, the matrix shape of the elements of $\overline{G}$ is very useful in describing their subgroups. Indeed, for the elements
\[
G_1=\left[\begin{array}{cc}
    Id_r  &  {H_1}_{r} \\
    0_r  & Id_r
\end{array}\right] \quad , \, \quad  G_2=\left[\begin{array}{cc}
    Id_r &  {H_2}_{r} \\
    0_r & Id_r
\end{array}\right] \in \overline{G},
\]
we have
\[
G_1 . G_2 =\left[\begin{array}{cc}
    Id_r  &  \left({H_1} + {H_2}\right)_{r} \\
    0_r  & Id_r
\end{array}\right]\in \overline{G} \quad \mbox{and} \quad G_1 . G_1 ^{-1} = Id_n, \mbox{ since} \quad G_1 ^{-1} = \left[\begin{array}{cc}
    Id_r  &  {- H_1}_{r} \\
    0_r  & Id_r
\end{array}\right].
\]
Hence, the product operation of the matrices in $\overline{G}$ may be reduced to the sum operation of the matrices in $\mathbb{F}_q ^{r \times r}$. Moreover, if we want to describe a subgroup of $\overline{G}$ generated by the elements
\[
\left[\begin{array}{cc}
    Id_r & {H_1}_r \\
     0_r & Id_r
\end{array}\right],  \left[\begin{array}{cc}
    Id_r & {H_2}_r \\
     0_r & Id_r
\end{array}\right],\ldots ,  \left[\begin{array}{cc}
    Id_r & {H_L}_r \\
     0_r & Id_r
\end{array}\right],
\]
then we need to describe the corresponding additive subgroup (or $\mathbb{F}_q$-vector subspace of $\mathbb{F}_q ^{r \times r}$) whose generators are ${H_1}_r , {H_2}_r ,\ldots ,{H_L}_r$.

From the previous consideration together with Theorem~\ref{distanciaminimaviarank}, we use the rank-metric code construction to provide a systematic way to obtain new Abelian non-cyclic orbit codes with larger cardinality than that obtained by the constructions of the known Abelian orbit codes. In this paper we consider Delsarte's matrix representation of MRD codes~\cite{delsarte}, as reported by Gabidulin in~\cite{stateofart}.

\begin{example}\label{exampleemf3}
Let $V= rs\left(\left[\begin{array}{cccccc}
     1&0&0&1&2&0  \\
     0&1&0&1&0&0  \\
     0&0&1&0&2&1
\end{array}\right]\right)$ be a $3$-dimensional vector subspace of $\mathcal{G}_3 (6,3)$. Given the $3 \times 3$-matrices
\begin{eqnarray}
H_1 &=&\left[\begin{array}{ccc}
     1&0&0  \\
     0&1&0  \\
     0&0&0
\end{array}\right], H_2 =\left[\begin{array}{ccc}
     0&0&0  \\
     0&1&0  \\
     0&0&1
\end{array}\right], H_3 =\left[\begin{array}{ccc}
     0&0&1  \\
     0&1&0  \\
     0&1&0
\end{array}\right], H_4 =\left[\begin{array}{ccc}
     0&0&2  \\
     2&0&0  \\
     0&1&0
\end{array}\right],\nonumber \\ H_5 &=&\left[\begin{array}{ccc}
     1&1&2  \\
     0&1&2  \\
     2&0&1
\end{array}\right] \mbox{ and } H_6 =\left[\begin{array}{ccc}
     0&0&0  \\
     0&0&1  \\
     2&1&1
\end{array}\right],
\end{eqnarray}
the rank-metric code $\mathcal{C}=\left\langle H_1 , H_2 , H_3, H_4, H_5, H_6 \right\rangle$ is an MRD code according to the Singleton bound, since its cardinality is equal to $729$ with minimum rank distance $d_R (\mathcal{C})=2$. Now, take the orbit code $C_{\overline{G}} (V)$, with
\begin{eqnarray}
\overline{G} &=&\left\langle \left[\begin{array}{cc}
    Id_3 & H_1  \\
    0_3 & Id_3
\end{array}\right],\left[\begin{array}{cc}
    Id_3 & H_2  \\
    0_3 & Id_3
\end{array}\right],\left[\begin{array}{cc}
    Id_3 & H_3  \\
    0_3 & Id_3
\end{array}\right], \left[\begin{array}{cc}
    Id_3 & H_4  \\
    0_3 & Id_3
\end{array}\right],\left[\begin{array}{cc}
    Id_3 & H_5  \\
    0_3 & Id_3
\end{array}\right]\right.,\nonumber \\
&&\left.\left[\begin{array}{cc}
    Id_3 & H_6  \\
    0_3 & Id_3
\end{array}\right] \right\rangle< GL_6 (3).
\end{eqnarray}
Thus, $C_{\overline{G}} (V)$ is a $(6,729,4,3)$-(ternary) constant-dimension code. Note that $\left|C_{\overline{G}} (V)\right|$ is close to the best lower bound on $A_3(6,4,3)$, which is 754. This and several other lower and upper bounds on $A_q (n,d,k)$ can be seen in http://subspacecodes.uni-bayreuth.de (see~\cite{table}). Finally, the largest known Abelian orbit code with the same length, distance and dimension, which is cyclic, has $364$ codewords, which is less than one-half of the number of codewords of $C_{\overline{G}}(V)$.
\end{example}
%

Considering a restriction involving the parameters $q$ and $n$, we state that the Abelian non-cyclic orbit codes described in this paper are better than any other orbit code construction known so far, even the non-Abelian cases. In fact, the non-Abelian orbit codes with generating group $\langle \alpha \rangle \rtimes \langle \sigma \rangle$~\cite{iran} have cardinality upper bounded by $n\left(\frac{q^n -1}{q-1}\right)$ and the minimum subspace distance less than or equal to $2k-2$. To the best of our knowledge, such a construction leads to the best orbit codes (It is worth mentioning that the construction like the one shown in ~\cite{china} leads to cyclic codes, which are not necessarily orbit codes). If there is an $\left(n,n\left(\frac{q^n -1}{q-1}\right),n-2,\frac{n}{2} \right)$-non-Abelian orbit code, then the $\left(n,q^n ,n-2,\frac{n}{2}\right)$-Abelian non-cyclic orbit codes being proposed are always better for $q-1\geq n$, since $q^n$ is less than $n\left(\frac{q^n -1}{q-1}\right)$, then $q^n <n\left(\frac{q^n -1}{q-1}\right)\leq(q-1)\left(\frac{q^n -1}{q-1}\right)\leq q^n -1$, a contradiction.
%

\section{Geometrically Uniform Subspace Codes}\label{codsubgeomunif}
From the classical coding theory, the class of geometrically uniform (GU) codes as proposed by Forney in~\cite{forney}, encompasses the Slepian group codes~\cite{codigodegrupo} and the lattices codes~\cite{sloane}, and its importance is due to the inherent richness of its algebraic and geometric structures.

Let $(M,d)$ be a metric space, with $M$ describing the ambient space and $d$ a metric. We want to emphasize that the GU codes in consideration may belong to more general spaces other than the Euclidean space.

\begin{definition}\label{codigosgeometricamenteuniformes}
Let $(M,d)$ be a metric space and $C \subset (M,d)$. Then $C$ is a GU code if, given two codewords $c_1$ and $c_2$ in $C$, there exists an isometry $u_{c_1 , c_2}$ such that $u_{c_1 , c_2}$ maps $c_1$ to $c_2$ while leaving $C$ invariant.

\begin{equation}
u_{c_1 , c_2}\left(c_1 \right)=c_2 \quad \mbox{ and }\quad u_{c_1 , c_2}(C)=C.
\end{equation}
\end{definition}

Since $C$ is GU, it follows that there exists a symmetry group $\Gamma(C)$ which acts transitively on $C$, i.e., given any $c \in C$, $C$ may be defined as the orbit code
\begin{equation}\label{orbita}
C=\left\{u(c) : u \in \Gamma(C)\right\}.
\end{equation}
%
%
\begin{definition}\label{groupominimo}
A generating group $G$ of $C$ is a subgroup of the symmetry group $\Gamma(C)$ that is minimally sufficient to generate $C$ from any arbitrary codeword $c \in C$. That is, if $G$ is a generating group of $C$, and $c \in C$, then $C$ is the orbit of $c$ under $G$, $C=\left\{u\left(c \right) : u \in G\right\}$, and the map $m:G\rightarrow C$ defined by $m(u)=u\left(c \right)$ is one-to-one.
\end{definition}

If $C$ is GU with generating group $G$ acting transitively on $C$, 
 then $C$ will be denoted by $C_G (c)$.

\begin{remark}
According to~\cite{zhewan}, for codes defined in the usual Euclidean metric space $\left(\mathbb{R}^n , d_E \right)$,
with $d_E$ denoting the Euclidean metric, there exists an equivalence between codes matched to groups~\cite{loeliger} and GU codes. This equivalence may be also extended naturally to the metric space $\left(\mathcal{G}_q (n,k), d_S \right)$.
\end{remark}
%

\begin{definition}
A \emph{Voronoi region} $R_V \left(c_1 \right)$ associated with any codeword $c_1\in C_G (c) \subseteq (M,d) $ is the set of all points in $M$ that are at least as close to $c_1$ as to any other codeword $c_2 \in C_G (c)$
\begin{equation}
\displaystyle{R_V \left(c_1 \right)=\left\{x\in M : d\left(c_1 ,x \right) = \min_{c_2 \in C_G (c)} d\left(c_2 ,x \right) \right\}}.
\end{equation}
\end{definition}

\begin{definition}
The \emph{global distance profile} $DP\left(c_1 \right)$ associated with any codeword $c_1 \in C_G (c) \subseteq (M,d)$ is the set of distances to all other codewords of $C_G (c)$.
\begin{equation}
DP\left(c_1\right)=\left\{ d\left(c_1 , c_2 \right) , c_2 \in C_G (c) \right\}.
\end{equation}
\end{definition}

As a consequence of these two previous definitions, one important result from \cite{forney} is the concept of geometrical uniformity.

\begin{theorem}\label{geounif}~\cite{forney}
If $C_G (c)$ is a GU code in $(M,d)$, then
\begin{itemize}
\item[(i)] All the Voronoi regions $R_V \left(c_1 \right)$ have the same shape, and indeed $R_V (c_2 )=$\linebreak $u_{c_1 , c_2} \left[R_V (c_1)\right]$, with $u_{c_1 , c_2}$ any isometry that takes $c_1$ to $c_2$,
\item[(ii)] The global distance profile $DP\left(c \right)$ is the same for all $c \in C_G (c)$, and it is denoted by $DP\left(C_G (c) \right)$.
\end{itemize}
\end{theorem}

From Theorem~\ref{geounif}, the regular arrangement of the codewords of the GU codes is due to the transitive action of $G$ on them. In addition, given a Voronoi region, the remaining Voronoi regions may be obtained by the group action, and consequently, all of them have the same shape and properties.

Let $C_G (c)$ be a GU code and  $H\lhd G$ a normal subgroup of $G$. A geometrically uniform partition was defined by Forney~\cite{forney} as a partition of $C_G (c)$ generated by the factor group $G/H=\left\{ Hg_1 ,  H g_2 ,..., Hg_t \right\}$, with
 $t=\left|G/H\right|$ and $g_1$ the identity element of $G$. Thus, by using a coset of $G/H$, a subcode $C_{H}\left(g_i c\right)$ is defined as follows.
\begin{equation}\label{defpartitiongu}
C_{H}\left(g_i c\right)=C_{Hg_i } (c):=\left\{h \left(g_i  (c)\right) : h\in H \right\},
\end{equation}
and $C_{G/H}:=\left\{C_{Hg_1} (c), C_{Hg_2} (c), ..., C_{Hg_t} (c) \right\}$ such that $\displaystyle{C_G (c)=\bigcup_{i=1} ^{t} C_{ Hg_i} (c)}$.

\begin{theorem}~\cite{forney}\label{partigeomunif}
Let $C_{G/H} (c)=\left\{C_{Hg_1} (c) ,  C_{Hg_2} (c) ,..., C_{Hg_t} (c) \right\}$ be a geometrically uniform partition of $C_{G} (c)$. Then the subcodes $C_{Hg_i} (c)$ of $C_G (c)$ in this partition are geometrically uniform, mutually congruent, and have $H$ as a common generating group.
\end{theorem}
%
%
%

From Definition~\ref{codigosgeometricamenteuniformes}, in order to obtain a complete understanding of geo\-me\-tri\-cal\-ly uniform subspace codes, the set of isometries acting on the projective space $\mathcal{P}_q (n)$ must be established~\cite{isometriasana}. This result is based on the famous Fundamental Theorem of Projective Geometry~\cite{baer} and \cite{teoriadegrupos}. Before presenting it, let us introduce two essential lemmas.

\begin{lemma}\cite{isometriasana}
If $\lambda: \mathcal{P}_q (n) \rightarrow \mathcal{P}_q (n)$ is an isometry, then $\lambda(\{0\}) \in \left\{\{0\}, \mathbb{F}_q ^n\right\}$.
\end{lemma}

\begin{lemma}\cite{isometriasana}\label{zeroazero}
Let $\lambda$ be an isometry and $V \in \mathcal{P}_q (n)$ arbitrary. If $\lambda(\{0\})=\{0\}$, then
\begin{equation}\label{equacao1111}
dim (V)= d_S (\{0\},V)=d_S (\{0\},\lambda(V))=dim\lambda(V).
\end{equation}
Otherwise, $\lambda(\{0\})=\mathbb{F}_q ^n$. As a consequence,
\begin{equation*}
dim (V)= d_S (\{0\},V)=d_S \left(\mathbb{F}_q ^n ,\lambda(V)\right)=n- dim\lambda(V).
\end{equation*}
\end{lemma}

To preserve the dimension of the codewords and, consequently, for the computation of the minimum distance, we suppose that all isometries satisfy the condition shown in~\eqref{equacao1111}. Thus, the characterization of the isometries acting on $\mathcal{P}_q (n)$ is shown next.

\begin{theorem}~\cite{isometriasana}\label{caractiisometrias}
Every isometry $\lambda$ acting on $\mathcal{P}_q (n)$, for $n>2$, $dim(V)=dim(\lambda(V))$ and any $V \in \mathcal{P}_q (n)$, is induced by a semilinear transformation $\left(A,\sigma_i \right)\in P\Gamma L_n \left(\mathbb{F}_q \right)$, such that $P\Gamma L_n \left(\mathbb{F}_q \right):=\left(GL_n \left(\mathbb{F}_q \right) /Z_n \left(\mathbb{F}_q \right)\right) \rtimes Aut\left(\mathbb{F}_q \right)$ is the projective semilinear group, with $Z_n \left(\mathbb{F}_q \right)$ the subgroup of scalar matrices and $Aut\left(\mathbb{F}_q \right)$ the group of automorphisms of $\mathbb{F}_q$.
\end{theorem}

The next result, Corollary~\ref{caracterizacaodasisometrias}, is a characterization of all isometries acting on $\mathcal{P}_q (n)$.

\begin{corollary}\cite{isometriasana}\label{caracterizacaodasisometrias}
Every isometry $\lambda$ acting on $\mathcal{P}_q (n)$, for $n > 2$, $dim(V) = dim(\lambda(V))$ and any
$V \in \mathcal{P}_q (n)$, is induced by a semilinear transformation $(A,\varphi) \in P\Gamma L_n \left(\mathbb{F}_q \right)$.
\end{corollary}

As it is shown in~\cite{isometriasana}, we can extend the action of $GL_n \left(\mathbb{F}_q \right)$ to $P\Gamma L_n \left(\mathbb{F}_q \right)$ in order to define orbit codes $C_G (V)$ such that $G\leq P\Gamma L_n \left(\mathbb{F}_q \right)$.
%

After the classification of the isometries in $\mathcal{P}_q (n)$, as provided by Corollary~\ref{caracterizacaodasisometrias}, we are able to characterize all geometrically uniform subspace codes in $\mathcal{G}_q (n,k)$, as follows. 

\begin{proposition}\label{codigodeorbitaegu}
Given $C\subseteq \mathcal{G}_q (n,k)$, $C$ is a geometrically uniform subspace code if, and only if, $C= C_G (V)$ is an orbit code, with $G \leq P\Gamma L_n \left(\mathbb{F}_q \right)$.
\end{proposition}

\begin{proof}
Given $G \leq P\Gamma L_n \left(\mathbb{F}_q \right)$ and $V$ a $k$-dimensional vector subspace of $\mathbb{F}_q ^n$, according to Definition~\ref{codigosgeometricamenteuniformes} (Equation~\eqref{orbita}), every geometrically uniform code in $\mathcal{G}_q (n,k)$ is an orbit code $C_G (V)$. Conversely, if $C_G (V)$ is an orbit code then, by Corollary~\ref{caracterizacaodasisometrias}, the elements of $G$ act as isometries on $\mathcal{P}_q (n)$. In particular, $G$ acts as a symmetry group on $C_G (V)$, because for any distinct codewords $Vg_i ,Vg_j \in C_G (V)$, the group structure assures the existence of a symmetry ($g_{i} ^{-1} g_j \in G$) that takes $Vg_i$ to $Vg_j$, for any $0 \leq i \leq j \leq |G|-1$. This is exactly the definition of GU codes.
\end{proof}

Given $\alpha$ a primitive element of $\mathbb{F}_{q^n}$, for $n>2$, consider the constant dimension code $C=C_{\langle \alpha \rangle} (V) \cup C_{\langle \alpha \rangle} (\sigma(V))$ such that $\sigma \in Aut\left(\mathbb{F}_{q^n} \right)\setminus\{Id_n \}$, with $\sigma(x)=x^q$, for any $x\in \mathbb{F}_{q^n}$. This code is a cyclic subspace code~\cite{etzion}, since it is closed to cyclic-shift of $\alpha$, but it is not an orbit code.

\begin{corollary}\label{cicliconaoegu}
Cyclic codes are geometrically uniform subspace codes if, and only if, they are orbit codes.
\end{corollary}

From now on, we write GUSC to refer to geometrically uniform subspace codes.

\begin{example}\label{exampleperfildist}
Let $\alpha$ be a root of the primitive polynomial $f(x)=x^4 +x +1 \in \mathbb{F}_2 [x]$, with $\displaystyle{\alpha \in \mathbb{F}_{2^4} \simeq \mathbb{F}_2 [x]/\langle f(x)\rangle}$. Moreover, given $V_1 =\left\{ 0,1, \alpha ,\alpha^4 \right\}$, a $2$-dimensional vector \linebreak subspace of $\mathbb{F}_{2^4}$, consider the cyclic orbit code $C_{\left\langle \alpha \right\rangle} \left(V_1 \right)$, with $\left|C_{\left\langle \alpha \right\rangle} \left(V_1 \right)\right|=15$.
%
%
%
%
As \linebreak$\displaystyle{\mathcal{G}_2 (4,2)= \bigcup_{i=1} ^3 C_{\langle \alpha \rangle} \left(V_i \right)}$, with $V_2=\left\{0,1, \alpha^2, \alpha^8 \right\}$ and $V_3 =\mathbb{F}_{2^2} = \left\{0,1, \alpha^5, \alpha^{10} \right\}$, then the Voronoi region of the codeword $V_1$ is
\begin{eqnarray*}
R_V \left(V_1 \right)&=&\left\{\alpha V_1 ,\alpha^3 V_1, \alpha^4 V_1, \alpha^{11} V_1 , \alpha^{12} V_1 ,\alpha^{14} V_1 , V_2 ,  \alpha V_2 , \alpha^{2} V_2 , \alpha^{4} V_2 , \alpha^{7} V_2 , \alpha^{8} V_2 , \alpha^{11} V_2 ,\right.\\
&&\left. \alpha^{13} V_2, \alpha^{14} V_2 , V_3 , \alpha V_3 ,  \alpha^{4} V_3 \right\}.
\end{eqnarray*}

The element $\alpha^{11}$ acts as a symmetry on $C_{\langle \alpha \rangle} \left(V_1 \right)$ and according to Theorem~\ref{geounif}, the Voronoi region $u_{V_1 , \alpha^{11} V_1 } \left[R_V \left(V_1 \right)\right]=R_V \left(\alpha^{11} V_1 \right)$ of the codeword $\alpha^{11} V_1$ is
\begin{eqnarray*}
R_V \left(\alpha^{11} V_1 \right)&=&\left\{\alpha^{12} V_1 , \alpha^{14} V_1 , V_1 ,\alpha^{7} V_1 , \alpha^{8} V_1 , \alpha^{10} V_1 , \alpha^{11} V_2 , \alpha^{12} V_2 ,\alpha^{13} V_2 , V_2 , \alpha^{3} V_2 , \alpha^4 V_2 ,\right.\\
&&\left.  \alpha^7 V_2 , \alpha^9 V_2 , \alpha^{10} V_2, \alpha V_3 , \alpha^{2} V_3 , V_3 \right\}.
\end{eqnarray*}
\end{example}

According to Theorem~\ref{partigeomunif}, a partition of a GU code is directly related with the structure of the generating group and, as it will be discussed ahead, from this partition it is possible to obtain some results related with, for instance, the reduced number of computations to obtain the minimum subspace distance of a Abelian GUSC. Thus, for completeness regarding geometrically uniform partitions, a result from group theory related to normal subgroups of $GL_n \left(\mathbb{F}_q \right)$ is required. The notation of the next theorem has been slightly modified in order to fit properly to the case in consideration since the original statement of the theorem encompasses a more general situation other than the general linear groups over finite fields.

\begin{theorem}\cite{matrixgroup}\label{caracterizacaodossubgruposnormais}
Given $n>1$, then every subgroup of $GL_n \left(\mathbb{F}_q \right)$ that contains $SL_n \left(\mathbb{F}_q \right)$ (the special linear group), or it is contained in $Z_n \left(\mathbb{F}_q \right)$, is a normal subgroup of $GL_n \left(\mathbb{F}_q \right)$. If $n>2$, or $n=2$, but $q\neq 2$ or $q\neq 3$, then any normal subgroup of $GL_n \left(\mathbb{F}_q \right)$ contains $SL_n \left(\mathbb{F}_q \right)$, or is contained in the center of $GL_n \left(\mathbb{F}_q \right)$, which is exactly $Z_n \left(\mathbb{F}_q \right)$.
\end{theorem}
%

\begin{definition}
A \emph{normal series} of a group $G$, with $e$  the identity element of $G$, is a sequence of subgroups
\begin{equation}
G=G_{0} \geq G_{1} \geq ... \geq G_{m-1} \geq G_m =\{e\}
\end{equation}
such that $G_{i+1}\vartriangleleft G_i$, for all $i=0,...,m-1$.

A \emph{composition series} is a normal series  
%
such that, for all $i=0,...,m-1$, either $G_{i+1}$ is a maximal normal subgroup of $G_i$, or $G_{i+1}=G_i$.
\end{definition}

%

If $H\lhd G < GL_n \left(\mathbb{F}_q \right)$ and $|G/H|=t$, then the \emph{geometrically uniform partition}~\eqref{defpartitiongu} of $C_G (V)$ induced by $H$ can be seen as a union of orbit subcodes. Indeed,
\begin{equation}
C_{G/H} (V) := \left\{C_{H g_1 } (V) , C_{ H g_2} (V) ,..., C_{ H g_{t}} (V)\right\} = \left\{C_{H} (V_1) , C_{ H} \left(V_2 \right) ,..., C_{H} \left(V_{t} \right)\right\},
\end{equation}
with $V_i = rs(g_i \mathcal{V})$ the $k$-dimensional subspaces of $\mathbb{F}_q ^n$, for $i=1,...,t$ and $g_1 =Id_n $. By Theorem~\ref{partigeomunif}, the orbit subcodes of $C_{G/H} (V)$ are mutually congruent.

Now, we consider some concepts from~\cite{biglieri} which will be adapted to our approach in order to reduce the number of computations to obtain the minimum subspace distance of Abelian orbit codes.

Let $B \subseteq \mathcal{P}_q (n)$ be a set. We define the \emph{intradistance set} $D_S (B)$ as the multiset of all the subspace distances among pairs of subspaces of $B$, that is,
\begin{equation}
D_S (B):=\left\{d_S \left(V_1 , V_2\right) : V_1, V_2 \in B\right\}.
\end{equation}
%
%

If $B_1$ and $B_2$ are two disjoint subsets of $\mathcal{P}_q (n)$, the interdistance set $D_S \left(B_1 ,B_2 \right)$ is the multiset of all the subspace distances among subspaces of $B_1$ and $B_2$, i.e.,
\begin{equation}
D_S \left(B_1 , B_2 \right):=\left\{\left\{d_S \left(V_1 , V_2\right) : V_1 \in B_1 \mbox{ and } V_2 \in B\right\}\right\}.
\end{equation}

\begin{definition}\label{particaojustaaaaaa}
A partition $B_1, B_2, ..., B_m$ of a set $X\subseteq \mathcal{G}_q (n,k)$ is called \emph{fair} if, for each $1\leq i \neq j \leq m$, hold
\begin{itemize}
\item[(i)] $B_i \neq B_j$,
\item[(ii)] $\left|B_i \right| = \left|B_j \right|$ and
\item[(iii)] $D_S \left(B_i \right)=D_S \left(B_j \right)$.
\end{itemize}
\end{definition}

Given $H \lhd G$, Theorem~\ref{partigeomunif} states that all geometrically uniform partitions $C_{G/H} (V)$ provide fair partitions.

\begin{definition}\label{cadeiadeparticoes}
The chain partition of a set $X \subseteq \mathcal{G}_q (n,k)$ is called \emph{fair} if any two elements of the partition at the same level of the chain include the same number of vectors and have equal intradistance sets.
\end{definition}

Given a normal series $G=G_0 \geq G_{1} \geq...\geq G_m \neq \left\{Id_n \right\}$ of $G<GL_n \left(\mathbb{F}_q \right)$, by successive applications of Theorem~\ref{partigeomunif}, we note that all partitions in different levels are fair and, therefore, we obtain a fair chain partition according to Definition~\ref{cadeiadeparticoes}.

\begin{example}
Let $p(x)=x^6 +x+ 1 \in \mathbb{F}_2 [x]$ be a primitive polynomial in $\mathbb{F}_2 [x]$, $\alpha $ a root of $p(x)$ such that $\displaystyle{\alpha \in \mathbb{F}_{2^6} \simeq \mathbb{F}_2 [x]/\langle p(x)\rangle}$ and $V=\{0,1,\alpha^8 , \alpha^{10}, \alpha^{20}, \alpha^{48}, \alpha^{59}, \alpha^{61}\}$ a $3$-di\-men\-si\-onal vector subspace of $\mathbb{F}_{2^6}$. From the composition series $\left\langle\alpha\right\rangle > \left\langle\alpha^3 \right\rangle > \left\langle\alpha^9 \right\rangle$, we obtain the following fair chain partition of $C_{\langle \alpha \rangle} (V)$
\begin{equation}
C_{\langle \alpha \rangle} (V) =  \bigcup_{i=0} ^2 C_{\langle \alpha^3 \rangle} \left(\alpha^i V \right)= \bigcup_{i=0} ^8 C_{\langle \alpha^9 \rangle} \left(\alpha^i V \right).
\end{equation}
%
\end{example}

\begin{definition}\label{defdopolparaintersetdist}
Given $H \lhd G$ and $g_i \in G$, let $C_H \left(V_i \right)=C_H \left(g_i V \right)$ be a subcode of $C_G (V)$. The \emph{distance profile} associated with $g \in G$ and $C_H \left(V_i \right)$ is represented by the following polynomial in the indeterminate $w$,
\begin{equation}
F\left(w,g,C_H \left(V_i \right)\right)=\sum_{d} a(d)w^d,
\end{equation}
with $a(d)$ the number of elements of $C_H \left(V_i \right)$ with subspace distance $d$ with respect to an element of $C_H \left(g V_i \right)=C_H \left(g g_i V \right)$.
\end{definition}

\begin{example}\label{examplesobreospolin}
Let $p(x)=x^6 +x+1 \in \mathbb{F}_2 [x]$ be a primitive polynomial. Let $\alpha$ be a root of $p(x)$, with $\alpha \in \mathbb{F}_{2^6} \simeq \displaystyle{\mathbb{F}_2 [x]/\langle p(x)\rangle}$ and $V=\left\{0,1,\alpha^{8},\alpha^{10},\alpha^{20},\alpha^{48},\alpha^{59},\alpha^{61}\right\}$ a $3$-dimensional subspace of $\mathbb{F}_{2^6}$. The cyclic orbit code $C_{\left\langle \alpha^3 \right\rangle} (V)$ can be partitioned as
\begin{equation}
C_{\left\langle \alpha^3 \right\rangle} (V) = \bigcup_{i=0} ^2 C_{\left\langle \alpha^9 \right\rangle} \left(\alpha^i V \right), \mbox{ with }
\end{equation}
\begin{eqnarray*}
C_{\langle \alpha^9 \rangle} (V)&:=&\left\{ \alpha^{9i} V : 0 \leq i \leq 6 \right\}=\left\{\left\{0,1,\alpha^{8},\alpha^{10},\alpha^{20},\alpha^{48},\alpha^{59},\alpha^{61}\right\},\right.\\
&&\left.\left\{0, \alpha^{5},\alpha^{7}, \alpha^{9}, \alpha^{17}, \alpha^{19}, \alpha^{29},\alpha^{57}\right\}, \left\{0, \alpha^{3},\alpha^{14},\alpha^{16},\alpha^{18},\alpha^{26},\alpha^{28},\alpha^{38}\right\},\right.\\
&&\left.\left\{0,\alpha^{12},\alpha^{23}, \alpha^{25}, \alpha^{27}, \alpha^{35}, \alpha^{37},\alpha^{47}\right\}, \left\{0, \alpha^{21},\alpha^{32},\alpha^{34},\alpha^{36},\alpha^{44},\alpha^{46},\alpha^{56}\right\},\right.\\
&&\left.\left\{0,\alpha^{2},\alpha^{30}, \alpha^{41}, \alpha^{43}, \alpha^{45}, \alpha^{53},\alpha^{55}\right\},\left\{0, \alpha,\alpha^{11},\alpha^{39},\alpha^{50},\alpha^{52},\alpha^{54},\alpha^{62}\right\}\right\},\\
C_{\langle \alpha^9 \rangle} \left(\alpha^3 V \right)&:=&\left\{ \alpha^{9i +3} V : 0 \leq i \leq 6 \right\}=\left\{\left\{0, \alpha,\alpha^{3}, \alpha^{11}, \alpha^{13}, \alpha^{23}, \alpha^{51},\alpha^{62}\right\},\right.\\
&&\left.\left\{0, \alpha^{8},\alpha^{10},\alpha^{12},\alpha^{20},\alpha^{22},\alpha^{32},\alpha^{60}\right\}, \left\{0,\alpha^{6},\alpha^{17}, \alpha^{19}, \alpha^{21}, \alpha^{29}, \alpha^{31},\alpha^{41}\right\},\right.\\
&&\left.\left\{0, \alpha^{15},\alpha^{26},\alpha^{28},\alpha^{30},\alpha^{38},\alpha^{40},\alpha^{50}\right\}, \left\{0,\alpha^{24},\alpha^{35}, \alpha^{37}, \alpha^{39}, \alpha^{47}, \alpha^{49},\alpha^{59}\right\},\right.\\
&&\left.\left\{0, \alpha^{5},\alpha^{33},\alpha^{44},\alpha^{46},\alpha^{48},\alpha^{56},\alpha^{58}\right\}, \left\{0,\alpha^{2},\alpha^{4}, \alpha^{14}, \alpha^{42}, \alpha^{53}, \alpha^{55},\alpha^{57}\right\}\right\}\\
C_{\langle \alpha^9 \rangle} \left(\alpha^6 V \right)&:=&\left\{ \alpha^{9i +6} V : 0 \leq i \leq 6 \right\}=\left\{\left\{0, \alpha^{2},\alpha^{4},\alpha^{6},\alpha^{14},\alpha^{16},\alpha^{26},\alpha^{54}\right\},\right.\\
&&\left.\left\{0,1,\alpha^{11}, \alpha^{13}, \alpha^{15}, \alpha^{23}, \alpha^{25},\alpha^{35}\right\}, \left\{0, \alpha^{9},\alpha^{20},\alpha^{22},\alpha^{24},\alpha^{32},\alpha^{34},\alpha^{44}\right\},\right.\\
&&\left.\left\{0,\alpha^{18},\alpha^{29}, \alpha^{31}, \alpha^{33}, \alpha^{41}, \alpha^{43},\alpha^{53}\right\}, \left\{0, \alpha^{27},\alpha^{38},\alpha^{40},\alpha^{42},\alpha^{50},\alpha^{52},\alpha^{62}\right\},\right.\\
&&\left.\left\{0,\alpha^{8},\alpha^{36}, \alpha^{47}, \alpha^{49}, \alpha^{51}, \alpha^{59},\alpha^{61}\right\}, \left\{0,\alpha^{5},\alpha^{7}, \alpha^{17}, \alpha^{45}, \alpha^{56}, \alpha^{58},\alpha^{60}\right\}\right\}.
\end{eqnarray*}

Note that the polynomials $F\left(w,\alpha^3 ,C_{\langle \alpha^9 \rangle} (V)\right)$ and $F\left(w,\alpha^6 ,C_{\langle \alpha^9 \rangle} (V)\right)$ are ob\-tai\-ned from the in\-ter\-dis\-tan\-ce sets $D_S \left(C_{\langle \alpha^9 \rangle} (V) , C_{\langle \alpha^9 \rangle} \left(\alpha^3 V \right) \right)$ and $D_S \left(C_{\langle \alpha^9 \rangle} (V) , C_{\langle \alpha^9 \rangle} \left(\alpha^6 V \right) \right)$, res\-pec\-ti\-vely, and
%
%
\begin{equation}
F\left(w,\alpha^3 ,C_{\langle \alpha^9 \rangle} (V)\right)=F\left(w,\alpha^6 ,C_{\langle \alpha^9 \rangle} (V)\right)= 7w^2 + 14w^4 +28w^6.
\end{equation}

The fact that the polynomials $F\left(w,\alpha^3 ,C_{\langle \alpha^9 \rangle} (V)\right)$ and $F\left(w,\alpha^6 ,C_{\langle \alpha^9 \rangle} (V)\right)$ are the same in Example~\ref{examplesobreospolin} is not a coincidence, as we will see in Lemma~\ref{cordoteodobigleiri}.

%
\end{example}

\begin{definition}
Given $H \lhd G$, such that $|G/H|=t$, the geometrically uniform partition $C_{G/H} (V)=\left\{C_H \left(g_1 V \right) , C_H \left(g_2 V \right) ,..., C_H \left(g_{t} V  \right)\right\}$ is called \emph{homogeneous} if the set \linebreak$\left\{F\left(w, g_i , C_H \left(g_j V \right)  \right)\right\}_{g_i \in G/H}$ does not depend on $C_H \left(g_j V \right)$. It is called \emph{strongly homogeneous} if $F\left(w, g_i , C_H \left(g_j V \right)  \right)$ does not depend on $C_H \left(g_j V \right)$, for any $g_i \in G/H$.
\end{definition}

\begin{theorem}~\cite{biglieri}\label{theoremdobiglieri}
If $G$ is an Abelian subgroup of $GL_n \left(\mathbb{F}_q \right)$, every geome\-trically uniform partition generated by subgroups of $G$ are strongly homogenous.
\end{theorem}
%
%

\begin{lemma}\label{cordoteodobigleiri}
Given $H$ a subgroup of $G$ and $C_G (V)$ an Abelian orbit code, let $C_{G/H} (V)=\left\{C_H \left(g_1 V \right) , C_H \left(g_2 V \right) ,..., C_H \left(g_{t} V  \right)\right\}$ be a geometrically uniform partition of $C_G (V)$. Then, for any $g_i \in G/H$, we have
\begin{equation}
F\left(w, g_i , C_H \left(V \right)\right)=F\left(w, g_i ^{-1} , C_H \left(V \right)\right).
\end{equation}
\end{lemma}

\begin{proof}
Each polynomial $F\left(w, g_i , C_H \left(V \right)\right)$ is computed from the interdistance set \linebreak$D_S \left(C_H \left(g_i V \right), C_H \left(V \right)\right)$. This set is described by
\begin{eqnarray*}
D_S \left(C_H \left(g_i V \right), C_H \left(V \right)\right)&=&\left\{\left\{d_S \left(h_j g_i V , h_k V\right)\right\}\right\} =\left\{\left\{d_S \left( g_i h_j  V , h_k V\right)\right\}\right\}\\
&=&\left\{\left\{d_S \left( h_j  V , g_i^{-1}  h_k V\right)\right\}\right\} =\left\{\left\{d_S \left( h_j  V , h_k g_i^{-1} V\right)\right\}\right\}\\
&=&D_S \left(C_H \left(g_i ^{-1} V \right), C_H \left(V \right)\right),
\end{eqnarray*}
for $h_j , h_k \in H$. As $D_S \left(C_H \left(g_i V \right), C_H \left(V \right)\right)=D_S \left(C_H \left(g_i ^{-1} V \right), C_H \left(V \right)\right)$, then the result follows.
\end{proof}

Lemma~\ref{cordoteodobigleiri} justifies why the polynomials $F\left(w,\alpha^3 ,C_{\langle \alpha^9 \rangle} (V)\right)$ and $F\left(w,\alpha^6 ,C_{\langle \alpha^9 \rangle} (V)\right)$ from Example~\ref{examplesobreospolin} are the same.

For an Abelian orbit code $C_G (V)$, the next theorem ensures that there is no need to compute all subspace distances $d_S \left(V,g_i V \right)$, for $2 \leq i \leq |G|$, in order to obtain the minimum subspace distance of this code.

\begin{theorem}\label{theoremprincioal}
Given $H$ a subgroup of $G$ and $C_G (V)$ an Abelian orbit code, let $C_{G/H} (V)$\linebreak$=\left\{C_H \left(g_1  V \right),C_H \left(g_2 V \right),...,C_H \left(g_{t} V \right)\right\}$ be a geometrically uniform partition of $C_G (V)$, with $C_H \left(g_1 V\right)=C_H \left(V \right)$, $G/H=\left\{g_1 ,g_2 ,...,g_{\frac{t}{2}},g_2 ^{-1},...,g_{\frac{t}{2}} ^{-1} \right\}$ and $I=\left\{2 , ..., {\frac{t}{2}} \right\}$. Then
\begin{equation}
d_S \left(C_G (V)\right)=\min_{i \in I} \left\{D_S \left(\{V\}, C_H \left(g_i V\right)\right)\right\}.
\end{equation}
\end{theorem}

\begin{proof}
The minimum subspace distance of $C_G(V)$ is computed as
\begin{equation}
d_S \left(C_G(V)\right)=\min\left\{d_S \left(V, g_i V\right) : g_i \in G\setminus \left\{g_1 \right\}\right\}.
\end{equation}

By Theorem~\ref{partigeomunif}, this minimum subspace distance can also be computed as
\begin{equation}\label{mindasdistanciaintersubconjuntos}
d_S \left(C_G (V)\right)=\min \left\{d_S \left(C_H (V)\right),  \min_{g_i \in G/H \setminus \left\{g_1 \right\}} \left\{D_S \left(C_H (V), C_H \left(g_i V\right)\right)\right\}\right\}.
\end{equation}

From the sets $D_S \left(C_H (V), C_H \left(g_i V\right) \right)$, the minimum subspace distance of each interdistance set can be computed as
\begin{equation}\label{proporpo}
D_S \left(\{V\}, C_H \left(g_i V \right)\right),
\end{equation}
since, for any $h \in H$, the distance profile of $D_S \left(\{h V \}, C_H \left(g_i V\right)\right)$ is simply a permutation of the distance profile obtained in~\eqref{proporpo}.

According to Lemma~\ref{cordoteodobigleiri}, the polynomials $F\left(w,g_i , C_H (V)\right)$ and $F\left(w,g_i ^{-1} , C_H (V)\right)$ \linebreak are the same and, consequently, the interdistance sets $D_S \left(C_H \left(g_i V\right) , C_H (V)\right)$ and \linebreak $D_S \left(C_H \left(g_i ^{-1} V \right) , C_H (V)\right)$ are equal. Thus, equation~\eqref{mindasdistanciaintersubconjuntos} can be written as
\begin{equation}
d_S \left(C_G (V)\right)=\min \left\{d_S \left(C_H (V)\right), \min_{i \in I}\left\{D_S \left(\{V\}, C_H \left(g_i V\right)\right)\right\}\right\}.
\end{equation}

As $d_S \left(C_H (V) \right)\geq d_S \left(C_G (V)\right)$, then $\displaystyle{d_S \left(C_H (V) \right)\geq \min_{i \in I}\left\{D_S \left(\{V\}, C_H \left(g_i V \right)\right)\right\}}$. Therefore, it is enough to compute $\displaystyle{\min_{i \in I}\left\{D_S \left(\{V\}, C_H \left(g_i V \right)\right)\right\}}$ to obtain the minimum subspace distance of $C_G (V)$.
\end{proof}

\begin{remark}
Using the notation as in Theorem~\ref{theoremprincioal}, given $g \in G/H$, with $g \neq g_1$ and $g^2 = g_1$, then we consider $g\in I$.
\end{remark}
%

\begin{example}
Let $p(x)=x^6 +x +1 \in \mathbb{F}_2 [x]$ be a primitive polynomial and $\alpha \in \mathbb{F}_{2^6}\simeq \mathbb{F}_2 [x]/\langle p(x)\rangle$ a root of $p(x)$. Given $V=\left\{0, \alpha^0 , \alpha^{1} ,\alpha^{4} ,\alpha^{6} ,\alpha^{16} ,\alpha^{24} ,\alpha^{33} \right\}$ a $3$-dimensional vector subspace of $\mathbb{F}_{2^6}$, let us compute the minimum subspace distance of the cyclic orbit code $C_{\langle \alpha \rangle} (V)$.

As the order of $\alpha$ is $ord(\alpha)=63$, take $H= \left \langle\alpha^{9}\right \rangle$, with $|H|=7$. Then, the geometrically uniform partition $C_{G/H} (V)$ is described by
\begin{eqnarray*}
C_{G/H} (V)&=&\left\{C_H (V), C_{H} (\alpha V), C_{H} \left(\alpha^2 V \right) ,C_{H} \left(\alpha^3 V \right) ,C_{H} \left(\alpha^4 V \right) ,C_{H} \left(\alpha^5 V \right) , C_{H} \left(\alpha^6 V \right)  \right.\\
&&\left.  C_{H} \left(\alpha^7 V \right) , C_{H} \left(\alpha^{8} V \right) \right\}.
\end{eqnarray*}

According to Theorem~\ref{theoremprincioal}, we just need to compute $D_S \left(\{V\} , C_{H} (\alpha V)\right)$, \linebreak $D_S \left(\{V\} , C_{H} \left(\alpha^2 V\right)\right)$, $D_S \left(\{V\},  C_{H} \left(\alpha^3 V\right)\right)$ and $D_S \left(\{V\},  C_{H} \left(\alpha^4 V\right)\right)$. These distances are shown in Table \ref{tabe4}, in which we adopt the notation $\alpha^i$ to represent the vector subspace $\alpha^i V$.

\begin{table}[h!]\label{tabdopenultcod}
\centering
\begin{tabular}{c|ccccccc}
$d_S (.,.)$ & $\alpha^1$  & $\alpha^{10}$  & $\alpha^{19}$  & $\alpha^{28}$ & $\alpha^{37}$  & $\alpha^{46}$  & $\alpha^{55}$  \\ \hline
$\alpha^0$          & 4 & 4 & 6 & 6 & 6 & 4 & 6 \\ \hline
\\
$d_S (.,.)$ & $\alpha^2$  & $\alpha^{11}$  & $\alpha^{20}$  & $\alpha^{29}$ & $\alpha^{38}$  & $\alpha^{47}$  & $\alpha^{56}$  \\ \hline
$\alpha^0$          & 4 & 6 & 4 & 4 & 6 & 4 & 6 \\ \hline
\\
$d_S (.,.)$ & $\alpha^3$  & $\alpha^{12}$  & $\alpha^{21}$  & $\alpha^{30}$ & $\alpha^{39}$  & $\alpha^{48}$  & $\alpha^{57}$  \\ \hline
$\alpha^0$          & 4 & 4 & 6 & 4 & 4 & 4 & 4 \\ \hline
\\
$d_S (.,.)$ & $\alpha^4$  & $\alpha^{12}$  & $\alpha^{21}$  & $\alpha^{30}$ & $\alpha^{39}$  & $\alpha^{48}$  & $\alpha^{57}$  \\ \hline
$\alpha^0$          & 4 & 6 & 6 & 4 & 4 & 6 & 4
\end{tabular}
\caption{Interdistance sets $D\left(\{V\},C_H \left(\alpha^i V\right)\right)$, for $1 \leq i \leq 4$ }\label{tabe4}
\end{table}

Therefore, the minimum subspace distance of $C_{\langle \alpha \rangle} (V)$ is 4, the same minimum subspace distance obtained in~\cite[Example 1]{etzion}. Here, we just had to compute 28 distances to find this value, whereas by using the traditional method 63 distance computations are needed.
\end{example}

The number of computations needed to obtain the minimum subspace distance of cyclic orbit codes is given next.

\begin{corollary}\label{corolariodoteoprincipal}
Let $\alpha$ be a primitive element of $\mathbb{F}_{q^n}$ and $V \in \mathcal{G}_q (n,k)$. If $q^n -1 = r\cdot s$, given $\left\langle \alpha^r \right\rangle $ a subgroup of $\langle \alpha\rangle$, then the number of computations needed to obtain the minimum subspace distance of $C_{\langle \alpha \rangle} (V)$ is
\begin{equation}
\left\lfloor\frac{(r-1)}{2}\right\rfloor \cdot \left(\frac{s}{q-1}\right).
\end{equation}
\end{corollary}

\begin{proof}
Let us consider the geometrically uniform partition $C_{\langle \alpha \rangle / \left\langle \alpha^r \right\rangle} (V) =$\linebreak $\left\{C_{\left\langle \alpha^r \right\rangle} (V),..., C_{\left\langle \alpha^r \right\rangle} \left( \alpha^{r-1} V \right) \right\}$. From Theorem~\ref{theoremprincioal}, to obtain the minimum subspace distance of $C_{\langle \alpha \rangle} (V)$, we just need to compute the interdistance sets $D_S \left(\{V\} , C_{\left\langle \alpha^r \right\rangle} \left(\alpha V\right)\right), ... ,$\linebreak $D_S \left(C_{\left\langle \alpha^r \right\rangle} \left(\alpha^{\left\lfloor\frac{(r-1)}{2}\right\rfloor} V\right)\right)$. Since each subcode $C_{\left\langle \alpha^r \right\rangle} \left(\alpha^{i} V\right)$ has $\displaystyle{\frac{s}{q-1}}$ codewords, the result follows.
\end{proof}

\section{Applications of Geometrically Uniform Subspace Codes to Multishot Subspace Codes}\label{multishot}

So far, we consider the use of subspace codes for the channel proposed by K�tter and Kschischang~\cite{koetterk} only once. Multishot subspace coding, where the subspace channel is used more than once, has been proposed as an alternative to construct subspace codes with good rate and error correcting capability instead of increasing either the finite field size $q$ or the length $n$. The method for constructing multishot subspace codes proposed in~\cite{nobrega} is inspired by the so-called multi-level construction given by~\cite{calderbank} for block-coded modulation schemes, originally proposed by Imai and Hirakawa in~\cite{imai}.

Since we focus on the constant dimension codes, the Grassmannian $\mathcal{G}_q (n,k)$ follows naturally. From the group action of $G\leq GL_n \left(\mathbb{F}_q \right)$ on $\mathcal{G}_q (n,k)$ and, consequently, from the geometrically uniform partitions, it is possible to obtain a systematic way to describe both well-defined partitions in all levels and to reduce considerably the number of computations needed to obtain the minimum subspace distance of each subset in the different levels.

The multishot subspace codes based on a multi-level construction being considered in this paper can be found in~\cite{nobrega}.

\subsection{Multishot Subspace Codes}
The $m$\emph{-extension} of the projective space $\mathcal{P}_q (n)$, $\mathcal{P}_q (n) ^m$ is defined as the set of $m$-tuples of subspaces in $\mathcal{P}_q (n)$. The number of elements in $\mathcal{P}_q (n) ^m$ is given by $\left|\mathcal{P}_q (n) ^m \right| = \left|\mathcal{P}_q (n) \right|^m$. Moreover, the \emph{extended subspace distance} between two elements $\textbf{V}=\left(V_1 , V_2 , ..., V_m \right), \textbf{U}=\left(U_1 , U_2 , ..., U_m \right) \in \mathcal{P}_q (n) ^m$ is defined as
\begin{equation}
d_S (\textbf{V},\textbf{U})=\sum_{i=1} ^m d_S \left(V_i , U_i \right),
\end{equation}
with $d_S (.,.)$, in the right-hand side, the usual subspace distance defined for subspace codes. Indeed, the extended subspace distance is a metric accounting for the error weights occurred in each transmission. So, from this new metric space, an $m$-length multishot (block) subspace code $\mathcal{C}$ or just an $m$-\emph{shot subspace code} $\mathcal{C}$ over $\mathcal{P}_q (n)$ is a non-empty subset of $\mathcal{P}_q (n)^m$, where the minimum distance is computed as
\begin{equation}
d_S (\mathcal{C}) =\min \left\{d_S (\textbf{V},\textbf{U}): \textbf{V},\textbf{U} \in \mathcal{C}, \textbf{V} \neq \textbf{U}\right\}.
\end{equation}

Information about bounds for the size, rate and error control capability of $m$-shot subspace codes can be found in~\cite{nobrega}. It is still possible to associate such codes with $1$-shot codes in $\mathcal{P}_q (mn)$.

As a motivation to consider the multishot subspace codes based on a multi-level construction, it is provided in~\cite{nobrega} two simple multishot subspace code constructions, where the first one considers selecting some subspaces of $\mathcal{P}_q (n)$ according to a prescribed minimum subspace distance. The second construction is based on an injective labeling of the elements from $\mathcal{P}_q (n)$ to elements of $\mathbb{Z}_{|\mathcal{P}_q (n)|}$ and then looking up for the best block code in $\mathbb{Z}_{|\mathcal{P}_q (n)|} ^m$, when considering the Hamming distance, which is mapped back to a corresponding $m$-shot subspace code. The multishot subspace codes based on a multi-level construction yields better codes as shown in~\cite{nobrega}.

Before presenting the multishot subspace codes based on a multi-level construction, the next definition states what we mean by an $L$-level partition and, in particular, a nested $L-$level partition. See~\cite{calderbank} for more detailed information.

\begin{definition}\label{multileveldef}
An $L$\emph{-level partition} is a sequence of partitions $\Gamma_0$, $\Gamma_1$, $\Gamma_2 ,...,\, \Gamma_L$, with the partition $\Gamma_i$ being a
refinement of $\Gamma_{i-1}$ in the following sense: The $L$-level partition determines a rooted tree with $L + 1$ levels. The root is the signal constellation itself (namely, $\Gamma_0$), and the vertices at level $i$ are the subsets that constitute the partition $\Gamma_i$. A
vertex $y$ at level $i$ is joined to the unique vertex $x$ at level $i - 1$ containing $y$  and to every vertex $z$ at level $i + 1$ that is
contained in $y$ . The subsets that form the partition $\Gamma_L$ are the leaves of this tree. We shall only consider nested partitions
in which every subset at level $i$ is joined to the same number $p_{i+1}$ of subsets at level $i + 1$. However, we do allow the degree of
a vertex to vary from level to level. For every subset at level $i$, we use the numbers $0, 1, ... , p_{i+1}-1$ to label the edges from
that subset to the subsets at level $i + 1$. The subsets in the partition $\Gamma_L$ can then be labeled by paths $\left(a_1 ,..., a_L \right)$, $0\leq a_j \leq p_j -1, 1\leq j \leq L$, from the root to the corresponding leaf; more generally, the subsets in the partition $\Gamma_j$ can be
labeled by paths $\left( a_1 ,..., a_j \right)$.
\end{definition}

For the multishot subspace code construction we assume nested partitions up to a certain level.

The \emph{intrasubset subspace distance at level} $l$ is defined as
\begin{equation}
d_S (\Gamma_l ) = \min_{\mathcal{S} \in \Gamma_l } \left\{ d_S (U,V) : U, V \in \mathcal{S} , U\neq V\right\}\mbox{, with } 0 \leq l \leq L.
\end{equation}
It is worth mentioning that the intrasubset subspace distance of the leaves $\mathcal{S}\in \Gamma_L$  will be denoted as $\infty$.

Consider an $L$-level partition of $\Gamma_0 = \mathcal{G}_q (n,k)$ (or even $\mathcal{P}_q (n)$). Let us obtain a multishot subspace code with a prescribed minimum subspace distance $d$ based on this multi-level construction. Take $L' \leq L$ the minimum level with $d_S (\Gamma_l )\geq d$, for all $l\leq L'$, and such that all the partitions are nested up to this level.

\begin{definition}\label{codigoscomponentes}
An $L'$-level code $\mathfrak{C}=\left[\mathfrak{C}_1, \mathfrak{C}_2 , ... , \mathfrak{C}_{L'} \right]$ in which the component codes $\mathfrak{C}_i$ are traditional block codes, for all $1\leq i \leq L'$, and $\left\{0,1,2,...,p_i -1\right\}$ is the alphabet, is given by sequences $\left(a_1 ^k , a_2 ^k , ..., a_{L'} ^k\right)$, with $a_i ^k \in \mathfrak{C}_i$. The minimum Hamming distance $d_H \left(\mathfrak{C}_l \right)$ of the component codes must satisfy
\begin{equation}
\min \left\{d_S \left(\Gamma_{l-1} \right) \cdot d_H \left(\mathfrak{C}_{l} \right):1\leq l\leq L' \right\}\geq d .
\end{equation}
\end{definition}

From these previous definitions, we are able to describe the multishot subspace codes based on multi-level construction with minimum subspace distance $d$ as proposed in~\cite{nobrega}. An $m$-shot subspace code $\mathcal{C}\subseteq \mathcal{P}_q (n) ^m$ is obtained from the array consisting of $L'$ rows and $m$ columns, with the $l$-th row being represented by a codeword of the component code $\mathfrak{C}_l$. The $i$-th coordinate of a codeword of $\mathcal{C}$ is obtained as follows: consider the array $A$. Note that the $i$-th column of this array, denoted by $\left(a_{1,i},a_{2,i},...,a_{L',i}\right)^T$, describes a path in the partition tree starting from the root node $\Gamma _{0}$ and going up to the corresponding subset $\Gamma_{L'}$, for $1\leq i \leq m$.

\subsection{Advantages in Using GUSC to Construct Multishot Subspace Codes}

From the previous statements and discussions about multishot subspace codes using multilevel construction and assuming that only constant dimension codes are going to be considered in this paper, we propose to use the action of $G\leq GL_n \left(\mathbb{F}_q \right)$ on the alphabet $S \subseteq \mathcal{G}_q (n,k)$ in order to realize the partitions. We make use of the hypothesis that the stabilizers of $G$ acting on distinct vector subspaces have the same cardinality. This condition is necessary to ensure nested partitions. For instance, it is not possible to partition $\mathcal{G}_2 (6,3)$ in nested partitions concerning the action of the subgroup $G=\langle \alpha \rangle$ generated by the primitive element $\alpha$ of $\mathbb{F}_{2^6}$ over $\mathcal{G}_2 (6,3)$, since one subset in the first level of this partition has 9 elements $\left(\mbox{the spread code }C_{\langle \alpha \rangle} \left(\mathbb{F}_{2^3} \right)\right)$ and the remaining subsets have 63 elements. Hence, we may assume $k \nmid n$ as an example of the condition which ensures equal cardinality for all orbits.

Under the conditions shown in the previous paragraph, we list some advantages in using the action of $G\leq GL_n \left(\mathbb{F}_q \right)$ on $S \subseteq \mathcal{G}_q (n,k)$. According to the construction proposed in~\cite{nobrega}, it is needed to compute the intrasubset subspace distance in each level in order to decide which component codes (See Definition~\ref{codigoscomponentes}) will be used to obtain the prescribed minimum subspace distance $d$ of the multishot subspace code. Considering a composition series $G=G_0 >G_1 >...>G_m =\{e \} $, it is possible to define all nested $(m+2)$-level partition over $S$ (also a fair chain partition, see Definition~\ref{cadeiadeparticoes}) in a systematic way, where each level actually is a collection of geometrically uniform partitions from the previous level (See Theorem~\ref{partigeomunif}). From this same theorem, we may reduce considerably the number of computations needed to obtain the intrasubset subspace distance, since the geometrically uniform partition from each orbit code produces mutually congruent subcodes which means codes with the same minimum subspace distance. In particular, if $G$ is an Abelian subgroup of $GL_n \left(\mathbb{F}_q \right)$, from Corollary~\ref{corolariodoteoprincipal}, then the number of computations may be further reduced.

The following example makes explicit what was just mentioned.

\begin{example}\label{exemplomultishot}
Let us consider the set $S= \mathcal{G}_2 (6,3) \setminus C_{\langle \alpha \rangle}\left(\mathbb{F}_{2^3}\right)$ as our signal constellation/alphabet for a multishot subspace code. Following the notation from Definition~\ref{multileveldef}, $S=\Gamma_0$.

Let $G=\langle \alpha \rangle$ be the cyclic group generated by a primitive element of $\mathbb{F}_{2^6}$. Then the alphabet $S$ is partitioned by the action of $G$ as
\begin{equation}
S=\bigcup_{i=1} ^{6} C_{\langle \alpha \rangle \rtimes \langle \sigma\rangle} \left(V_i\right),
\end{equation}
with
\begin{eqnarray*}
V_1 &:=&\left\{0,1,\alpha, \alpha^{4}, \alpha^{6} , \alpha^{16} , \alpha^{24}, \alpha^{33}\right\},\\
V_2 &:=&\left\{0,1,\alpha, \alpha^2 , \alpha^{6} , \alpha^{7}, \alpha^{12} , \alpha^{26}\right\},\\
V_3 &:=&\left\{0,\alpha^{7}, \alpha^{16}, \alpha^{18}, \alpha^{28} , \alpha^{32} , \alpha^{49}, \alpha^{52}\right\},\\
V_4 &:=&\left\{0,\alpha, \alpha^{3}, \alpha^{12}, \alpha^{13} , \alpha^{18} , \alpha^{26}, \alpha^{48}\right\},\\
V_5 &:=&\left\{0,\alpha, \alpha^{18}, \alpha^{22}, \alpha^{29} , \alpha^{42} , \alpha^{43}, \alpha^{48}\right\},\\
V_6 &:=&\left\{0,\alpha^{4}, \alpha^{17}, \alpha^{26}, \alpha^{39} , \alpha^{54} , \alpha^{61}, \alpha^{62}\right\}.
\end{eqnarray*}

We call attention to the fact that $\Gamma_0$ consists of 22 cyclic orbit codes, using the group $\langle \alpha \rangle \rtimes \langle \sigma\rangle$ just to simplify the notation. All these 22 cyclic orbit codes have the same cardinality, since all the initial points have the same stabilizer. Note also that $\left|C_{\langle \alpha \rangle \rtimes \langle \sigma\rangle} \left(V_j \right)\right|=378$, for $j \in\{2,4\}$, $\left|C_{\langle \alpha \rangle \rtimes \langle \sigma\rangle} \left(V_j \right)\right|=126$, for $j \in\{1,5\}$, and $\left|C_{\langle \alpha \rangle \rtimes \langle \sigma\rangle} \left(V_j \right)\right|=189$, for $j \in\{3,6\}$.

%
%
%

Next, we set up two scenarios related to the partitions of $S$:
\begin{itemize}
\item[(i)] Consider the partitions given by the action of the groups on the composition series $\left\langle\alpha\right\rangle > \left\langle\alpha^3 \right\rangle > \left\langle\alpha^9 \right\rangle$. In this case, we apply Theorem~\ref{partigeomunif} and Corollary~\ref{corolariodoteoprincipal} which reduce the number of subsets/subcodes to be checked at each level and the number of calculations to obtain their minimum subspace distances. Consequently, the number of computations to obtain the intrasubset subspace distance of each level is reduced. In addition, partition levels are built in a well-structured manner, since we take into account the factor-group structure.

\item[(ii)] Consider the same partitions (and their respective subsets) of $S$ provided by $(i)$, but now looking at them only as a collection of vector subspaces not generated by the group action. Thus, since we do not refer to these subsets as orbits, we must calculate the minimum subspace distance of each subset as usual, that is, by taking the minimum subspace distance between all the different pairs of subspaces in that subset. In this case, if a subset has $x$ elements, $\left(\begin{array}{c} x\\ 2\end{array}\right)$ computations are required to obtain its minimum subspace distance, where $\left(\begin{array}{c} .\\ .\end{array}\right)$ denotes the usual binomial coefficient.
\end{itemize}

We will compare the number of computations required to obtain the intrasubset subspace distance in the $\Gamma_l$-levels from $S$ according to the scenarios $(i)$ and $(ii)$, for $1\leq l \leq 3$, since $\Gamma_0 = S$ and $\Gamma_4$ are the root and the leaves of the tree, respectively (See Definition~\ref{multileveldef}).

In $\Gamma_1$-level, $(i)$-scenario, from Corollary~\ref{corolariodoteoprincipal}, it is required $28\times 22=616$ computations in order to obtain the $\Gamma_1$-intrasubset subspace distance. In the $(ii)-$scenario, $\left(\begin{array}{c} 63\\ 2\end{array}\right)\times 22=42966$ computations are required to obtaining the same $\Gamma_1$-intrasubset subspace distance.

In $\Gamma_2$-level, $(i)$-scenario, each of the subsets/codes $C_{\langle \alpha \rangle}\left(V_i \right)$ from $\Gamma_1$, for $1\leq i \leq 22$, will be geometrically uniform partitioned in three cyclic orbit subcodes generated by $\left\langle\alpha^3 \right\rangle$. From Theorem~\ref{partigeomunif}, it is needed to compute just the minimum subspace distance of one of these three cyclic orbit subcodes, since they are mutually congruent. Again, using Corollary~\ref{corolariodoteoprincipal}, it is required $7$ computations to obtain its minimum subspace distance. Then, it is necessary to implement $7\times22=154$ computations to obtain the $\Gamma_2$-intrasubset subspace distance. In the $(ii)-$scenario, $\left(\begin{array}{c} 21\\ 2 \end{array}\right)\times66=13860$ computations are required to obtain the same $\Gamma_2$-intrasubset subspace distance.

In $\Gamma_3$-level, $(i)$-scenario, each of the $66$ cyclic orbit subcodes generated by $\left\langle\alpha^3 \right\rangle$ will be geometrically uniform partitioned by the action of the subgroup $\left\langle\alpha^9 \right\rangle$, totalizing $198$ subsets/cyclic orbit subcodes. Once more, from Theorem~\ref{partigeomunif}, $6\times22=132$ computations are required to obtain the $\Gamma_3$-intrasubset subspace distance. In the $(ii)-$scenario, $\left(\begin{array}{c} 7\\ 2 \end{array}\right)\times 198=4158$ computations are required to obtain the same $\Gamma_3$-intrasubset subspace distance.

The $\Gamma_4$-level refers to the leaves, namely, the unitary subsets formed by the words of $S$.
\end{example}

Therefore, Example~\ref{exemplomultishot} shows how the algebraic and geometric structures of GUSC may reduce the number of computations required to obtain the intrasubset minimum distance, and consequently, to improve implementation conditions for multishot subspace codes.

\begin{remark}
Obviously it is not required to implement all the computations stated in Example~\ref{exemplomultishot} at each level of the partition of $S=\mathcal{G}_2 (6,3) \setminus C_{\langle \alpha \rangle}\left(\mathbb{F}_{2^3}\right)$ to obtain their intrasubset subspace distances, since the possible minimum subspace distances are just $6,4$ and $2$. Such implementations become more interesting when considering larger alphabets.

In this example, our goal was to emphasize the systematic way to describe nested partitions provided by the composition series and the number of computations we may reduce in all non-trivial levels of such partitions using the geometrically uniform properties.
\end{remark}

\section{Conclusion}\label{section4}
In this paper, we have characterized the orbit codes as geometrically uniform codes and a new construction of Abelian non-cyclic orbit codes has been presented. From the geometric uniformity and the characterization of all normal subgroups of the general linear group over finite fields, we reinterpreted and made use of the concept of geometrically uniform partitions to orbit codes, which provide essential information about the algebraic and geometric structures. In particular, we analyzed partitions of the Abelian orbit codes, and from these partitions, a reduction in the number of computations needed to determine the minimum subspace distance was established in Theorem~\ref{theoremprincioal}. Furthermore, $L$-level partitions based on group actions yields a systematic way to describe it and, consequently, the number of computations needed to determine the intrasubset subspace distance is considerably reduced, optimizing multishot subspace code constructions.

\section*{Acknowledgment}
The authors would like to thank the financial support received from FAPESP under grant \linebreak 503891/2011-8, from CNPq under grant 303059/2010-9 and from CAPES and CNPq for the PhD scholarships. 

\end{document}